\documentclass{article}
\PassOptionsToPackage{numbers}{natbib}
\usepackage{optml}
\usepackage{amsmath}

\usepackage[utf8]{inputenc} 
\usepackage[T1]{fontenc}    
\usepackage{url}            
\usepackage{booktabs}       
\usepackage{nicefrac}       
\usepackage{microtype}      

\usepackage{graphicx,amsmath,amsfonts,amsthm,amscd,amssymb,bm,url,color,latexsym}
\usepackage{bbm}
\usepackage{algorithm}
\usepackage{algorithmicx}
\usepackage{algpseudocode} 
\usepackage{verbatim}
\usepackage{framed}
\usepackage{comment}
\usepackage{tikz}
\usepackage{fge}
\usepackage{mathtools}

\allowdisplaybreaks  

\usepackage[toc, page]{appendix}

\usepackage[capitalize]{cleveref}
\newtheorem{theorem}{Theorem}[section]
\newtheorem{lemma}[theorem]{Lemma}

\renewcommand{\mathbf}{\boldsymbol}

\newcommand{\mb}{\mathbf}
\newcommand{\mc}{\mathcal}

\newcommand{\bb}{\mathbb}
 
\newcommand{\set}[1]{\left\{ #1 \right\}}

\newcommand{ \brac }[1]{\left[ #1 \right]}
\newcommand{ \paren }[1]{ \left( #1 \right) }

\DeclareMathOperator{\mini}{minimize}

\DeclareMathOperator{\st}{subject\; to}

\newcommand{\norm}[2]{\left\| #1 \right\|_{#2}}

\newcommand{\innerprod}[2]{\left\langle #1,  #2 \right\rangle}

\newcommand{\expect}[1]{\bb E\left[ #1 \right]}

\numberwithin{equation}{section}

\title{A Local Analysis of Block Coordinate Descent for Gaussian Phase Retrieval}
\author{\name David Barmherzig \email{davidbar@stanford.edu}\\
  \addr{Stanford University}\\
  \name Ju Sun \email{sunju@stanford.edu}\\
  \addr{Stanford University}
}

\begin{document}

\maketitle

\begin{abstract}
While convergence of the Alternating Direction Method of Multipliers (ADMM) on convex problems is well studied, convergence on nonconvex problems is only partially understood. In this paper, we consider the Gaussian phase retrieval problem, formulated as a linear constrained optimization problem with a biconvex objective. The particular structure allows for a novel application of the ADMM. It can be shown that the dual variable is zero at the global minimizer. This motivates the analysis of a block coordinate descent algorithm, which is equivalent to the ADMM with the dual variable fixed to be zero. We show that the block coordinate descent algorithm converges to the global minimizer at a linear rate, when starting from a deterministically achievable initialization point.
\end{abstract}

\section{Introduction}
The Phase Retrieval (PR) problem consists of recovering a vector $\mb x \in \mathbb{R}^n$ (or $\mathbb{C}^n$) from a set of magnitude measurements $y_k=|\mb a_k^* \mb x|, k=1, \cdots, m$, where $\mb a_k \in \mathbb{R}^n$ (or $\mathbb{C}^n$), $k=1, \cdots, m$, are known as the \textit{measurement vectors}. PR arises in many physical settings (~\cite{harrison1993pr, walther1963, balan2010, shechtman2015phase}), in which case $\mb a_k$'s are derived from the Fourier basis vectors. Toward mathematical understanding, recent efforts have focused on the \textit{generalized phase retrieval} (GPR) problem, in which $\mb a_k$'s can be vectors other than Fourier. Numerous recovery results are now available for the case $\mb a_k$'s are Gaussian, as summarized in~\cite{jaganathan2015phase}. Among them are the results based on nonconvex optimization, mostly based on the following template: firstly, an initialization close to a global minimizer is found using a spectral method; secondly, a gradient descent type algorithm is shown to converge locally to a global minimizer when starting from the initialization. 

ADMM works remarkably well on certain structural convex problems and comes with strong convergence guarantees~\cite{boyd2011admm}. Empirically, ADMM also works surprisingly well\footnote{Indeed, the working algorithm used for practical Fourier PR can be seen as a variant of the ADMM~\cite{wen2012alternating}. } on certain structural nonconvex problems. However, the current theories (see, e.g., \cite{wang2015global}) only guarantee convergence to critical points, aka global convergence. In this work, we study the local convergence behavior of ADMM, working with the GPR as a model problem. On a natural constrained least-squares formulation for GPR under the Gaussian measurement model, the dual variable at the optimal point is shown to be zero. This in turn motivates us to analyze a block coordinate descent (BCD) algorithm, which is equivalent to ADMM with the dual variable fixed to be zero. In this preliminary study, we apply the BCD algorithm to the expected\footnote{...due to the randomness induced by the random measurement vectors.} optimization problem. We show that the BCD converges to a global minimizer with a linear converge rate when being initialized in a neighborhood of the global minimizers, aka local convergence. The required initialization can be obtained efficiently through a spectral decomposition, as shown in several prior works.

\section{A biconvex problem formulation and the ADMM algorithm}

Suppose $\mb x \in \mathbb{R}^n$, and consider the magnitude measurements $y_k = |\mb a^T \mb x|$, where $\mb a_k$'s are i.i.d standard Gaussian vectors. We introduce the following biconvex least-squares formulation of the Gaussian phase retrieval problem
\begin{align} \label{eqn:obj-finite}
\begin{split}
\mini_{\mb z, \mb w \in \mathbb{R}^n} & \quad f \paren{\mb z, \mb w} \doteq \frac{1}{4m} \sum_{k=1}^m \paren{y_k^2 - \mb a_k^T \mb z \mb a_k^T \mb w}^2 \\
\st  & \quad \mb z = \mb w.
\end{split}
\end{align}
This is the variable-splitting reformulation of another least-squares formulation 
\begin{align}
\mini_{\mb z \in \mathbb{R}^n} & \quad \frac{1}{4m} \sum_{k=1}^m \paren{y_k^2 - \paren{\mb a_k^T \mb z}^2 }^2, 
\end{align}
which has been systematically studied before (e.g.,~\cite{candes2015wf,Sun2017Geometric}).

Note that $f$ is biconvex (and in fact biquadratic) and that we have a linear equality constraint.  Following the approach of~\cite{boyd2011admm}, we consider the ADMM algorithm applied to~\cref{eqn:obj-finite}, which is given by
\begin{align} \label{alg:ADMM}
\begin{split}
\mb z^{(k+1)} & = \mathop{\arg\min}_{\mb z} \mc L\paren{\mb z, \mb w^{(i)}, \mb \lambda^{(k)}}, \\
\mb w^{(k+1)} & = \mathop{\arg\min}_{\mb w} \mc L\paren{\mb z^{(k+1)}, \mb w, \mb \lambda^{(k)}}, \\
\mb \lambda^{(k+1)} & = \mb \lambda^{(k)} + \rho\paren{\mb z^{(k+1)} - \mb w^{(k+1)}}, 
\end{split}
\end{align}
where $\rho >0$ is a chosen penalty parameter, and $L$ is the augmented Lagrangian function given by
\begin{align} \label{eqn:Lagrangian}
\mc L\paren{\mb z, \mb w, \mb \lambda} \doteq f \paren{\mb z, \mb w} + \innerprod{\mb \lambda}{\mb z - \mb w} + \frac{\rho}{2} \norm{\mb z - \mb w}{}^2.
\end{align}

\section{A natural reduction to block coordinate descent}

We observe the following key property of the augmented Lagrangian function~\cref{eqn:Lagrangian}.

\begin{lemma} \label{lem:Lag-crit}
A triple $\paren{\mb z, \mb w, \mb \lambda}$ is a critical point of $\mc L\paren{\mb z, \mb w, \mb \lambda}$ if and only if
\begin{align}
\partial_{\mb z} f = \partial_{\mb w} f = \mb 0, \quad \mb z = \mb w, \quad \mb \lambda = \mb 0. 
\end{align}
\end{lemma}

Since \cref{lem:Lag-crit} shows that $\mb \lambda=0$ at the global minima, we are motivated to consider a modified version of the ADMM algorithm \eqref{alg:ADMM}, in which the dual variable $\mb \lambda$ is fixed to equal zero.  This is equivalent to the block coordinate descent, or BCD, algorithm applied to the objective function $f(\mb z, \mb w) + \rho/2 \cdot  \| \mb z - \mb w\|^2$.  This is given by

\begin{align} \label{alg:block-coord-desc}
\begin{split}
\mb z^{(k+1)} & = \mathop{\arg\min}_{\mb z} f(\mb z, \mb w^{(k)}) + \frac{\rho}{2} \| \mb z - \mb w^{(k)}\|^2, \\
\mb w^{(k+1)} & = \mathop{\arg\min}_{\mb w} f(\mb z^{(k+1)}, w) + \frac{\rho}{2} \| \mb z^{(k+1)} - \mb w\|^2. 
\end{split}
\end{align}

\section{Linear Convergence}

In the spirit of providing a preliminary result, we shall consider the expected objective, which is given by
\begin{align} \label{eqn:exp-val}
g(\mb z, \mb w) & \doteq \expect{f(\mb z, \mb w)} + \frac{\rho}{2}\| \mb z - \mb w\|^2\\
&=\frac{3}{2}\norm{\mb x}{}^4 + \paren{\mb w^T \mb z}^2 + \frac{1}{2}\norm{\mb z}{}^2 \norm{\mb w}{}^2 - 2\mb x^T \mb z\mb x^T \mb w - \norm{\mb x}{}^2 \mb w^T \mb z + \frac{\rho}{2} \| \mb z - \mb w\|^2.
\end{align}


Secondly, we shall assume that our initial point $\paren{\mb z^{(0)}, \mb w^{(0)} }$ lies within the set

\begin{align}
N_{\mb x} \doteq \set{\paren{\mb z, \mb w}: \norm{\mb z - \mb x}{} \le \frac{1}{8} \norm{\mb x}{} \; \text{and}\; \norm{\mb w - \mb x}{} \le \frac{1}{8} \norm{\mb x}{}}.  
\end{align}

With sufficiently many samples, i.e., $m$ large enough, it is easy to obtain such initialization via a spectral decomposition. We record such a result proved in~\cite{candes2015wf}; see also a recent refinement that requires less samples~\cite{mondelli2017fundamental}. 

\begin{lemma} \label{lem:spec-init}
Suppose that $m \ge C_0 n \log n$, and let $\mb x^{(0)}$ be the top eigenvector of $Y=\sum_{k=1}^m y_k \mb a_k \mb a_k^T$ normalized such that $\norm{\mb x^{(0)}}{}^2 = \sum_{k=1}^m y_k/\sum_{k=1}^m \norm{\mb a_k}{}^2$. Then, $\|\mb x^{(0)}-\mb x\| \leq \frac{1}{8}\|x\|$. \footnote{Strictly speaking, either $\|\mb x^{(0)} - \mb x\| \le \frac{1}{8}\norm{\mb x}{}$ or $\|\mb x^{(0)} + \mb x\| \le \frac{1}{8}\norm{\mb x}{}$. Since the sign is not recoverable, recovering either $\mb x$ or $-\mb x$ is fine. We assume without loss of the generality the closeness to $\mb x$. } Here $C_0$ is an absolute constant. 
\end{lemma}

We next list several lemmas that are essential to obtaining our main local convergence result. All the proofs are deferred to the appendix. The next lemma says in the set $N_{\mb x}$, the objective $g$ is jointly strongly convex in $\paren{\mb z, \mb w}$. 
\begin{lemma} \label{lem:strong-cvx}
Suppose $\rho \ge \norm{\mb x}{}^2$. For all points $(\mb z, \mb w) \in N_{\mb x}$, 
\begin{align}
\nabla^2 g(\mb z, \mb w) \succeq \frac{1}{3} \|{\mb x}\|^2 \mb I. 
\end{align}
\end{lemma}
The next no-escape lemma ensures that for convergence analysis, we only have to deal with the set $N_{\mb x}$. 
\begin{lemma} \label{lem:no-escape}
Suppose $\rho \ge \frac{27}{8} \norm{\mb x}{}^2$. Then the BCD iterate sequence $\set{\paren{\mb z^{(k)}, \mb w^{(k}}}$ stays in $N_{\mb x}$. 
\end{lemma}
The next gradient Lipschitz result is essential to deriving a concrete convergence rate, similar to most other convergence proofs. 
\begin{lemma} \label{lem:block-Lipschitz}
$g(\mb z, \mb w)$ is block Lipschitz on $N_{\mb x}$.  More specifically, for all $(\mb z, \mb w) \in N_{\mb x}$ and all $\mb h_{\mb z}, \mb h_{\mb w} \in \mathbb{R}^n$,
\begin{align} \label{eqn:block-Lipschitz}
\|\nabla_{\mb z}g(\mb z + \mb h_{\mb z}, \mb w)-\nabla_{\mb z}g(\mb z, \mb w)\| & \leq L_{\mb z}\|\mb h_{\mb z}\|, \\
\|\nabla_{\mb w}g(\mb z, \mb w+\mb h_{ \mb w})-\nabla_{\mb z}g(\mb z, \mb w)\| & \leq L_{\mb w} \|\mb h_{\mb w}\|,
\end{align}
where $L_{\mb z}=L_{\mb w}=4\| \mb x\|^2 + \rho$.
\end{lemma}

It then follows from the Taylor theorem that (see, e.g.,~\cite{bertsekas1999nonlinear}, Proposition A.24), that
\begin{align} \label{eqn:block-descent-lemma}
g(\mb z+ \mb h_{\mb z}, \mb w) \le g(\mb z, \mb w) + \nabla_{\mb z}^T g(\mb z, \mb w) \mb h_{\mb z}+\frac{(4\| \mb x\|^2 + \rho)}{2}\| \mb h_{\mb z}\|^2, \\
g(\mb z, \mb w + \mb h_{\mb w}) \le g(\mb z, \mb w) + \nabla_{\mb w}^T  g(\mb z, \mb w)\mb h_{\mb w}+\frac{(4\| \mb x\|^2 + \rho)}{2}\| \mb h_{\mb w}\|^2, 
\end{align}
whenever $\paren{\mb z, \mb w}$ lie in $N_{\mb x}$.

We now state our main result.

\begin{theorem} \label{thm:lin-conv}
Consider $g(\mb z, \mb w)$ as given by \cref{eqn:exp-val}, and assume $\rho \ge \frac{27}{8}\norm{\mb x}{}^2$ and $(\mb z^{(0)}, \mb w^{(0)}) \in N_{\mb x}$.  Then, the block coordinate descent algorithm listed in~\cref{alg:block-coord-desc} converges linearly to the point $(\mb x, \mb x)$. Specifically, 
\begin{align}
\norm{\paren{\mb z^{(k)}, \mb w^{k}} - \paren{\mb x, \mb x}}{} \le \paren{1-\frac{\norm{\mb x}{}^2}{12\norm{\mb x}{}^2 + 3\rho}}^{k/2} \sqrt{\frac{6}{\norm{\mb x}{}} \brac{g\paren{\mb z^{(0)}, \mb w^{(0)}} - g\paren{\mb x, \mb x}}}. 
\end{align}
\end{theorem}

Convergence of BCD method on strongly convex function is well known. We adapt a proof appearing in~\cite{beck2013blockcd}.

\begin{proof}[Proof of Theorem~\ref{thm:lin-conv}]
\begin{align} \label{eqn:monot-eqn}
g(\mb z^{(k)}, \mb w^{(k)}) - g(\mb z^{(k+1)}, \mb w^{(k)}) &\ge  g(\mb z^{(k)}, \mb w^{(k)})- g\paren{\mb z^{(k)} - \frac{1}{L_{\mb z}}\nabla_{\mb z} g\paren{\mb z^{(k)}, \mb w^{(k)}}, \mb w^{(k)}} \nonumber \\
&\ge \frac{1}{2L_{\mb z}}\|\nabla_{\mb z}g(\mb z^{(k)}, \mb w^{(k)})\|^2 \nonumber \\
& = \frac{1}{2L_{\mb z}}\|\nabla g(\mb z^{(k)}, \mb w^{(k)})\|^2, 
\end{align}
where the first line follows as $\mb z^{(k+1)}$ minimizes $g(\mb z,\mb w^{(k)})$, the second line follows from~\cref{eqn:block-descent-lemma}, and the last line follows as $\nabla_{\mb w} g\paren{\mb z^{(k)}, \mb w^{(k)}} = \mb 0$.  By \cref{lem:strong-cvx}, $g$ is strongly convex, and hence
\begin{align*} 
g(\mb z_2, \mb w_2) \ge g(\mb z_1, \mb w_1) + \langle \nabla g(\mb z_1, \mb w_1), (\mb z_2-\mb z_1, \mb w_2-\mb w_1) \rangle+ \frac{\sigma}{2}\|(\mb z_2-\mb z_1, \mb w_2-\mb w_1)\|^2,
\end{align*}
for all $(\mb z_1, \mb w_1), (\mb z_2, \mb w_2) \in N_{\mb x}$, where $\sigma=\frac{1}{3} \|{\mb x}\|^2$.  Minimizing both sides w.r.t. $(\mb z_2, \mb w_2)$, we have 
\begin{align} \label{eqn:global-strong-cvx-eqn}
g(\mb z_1, \mb w_1)-g(\mb x, \mb x) \le \frac{1}{2\sigma}\|\nabla g(\mb z_1, \mb w_1)\|^2,
\end{align}
for all $(\mb z_1, \mb w_1) \in N_{\mb x}$.  It then follows from \cref{eqn:monot-eqn,eqn:global-strong-cvx-eqn} that
\begin{align}
g(\mb z^{(k)}, \mb w^{(k)})-g(\mb x,\mb x) & \le \frac{1}{2 \sigma}
\| \nabla g(\mb z^{(k)}, \mb w^{(k)})\|^2 \\
& \le \frac{L_{\mb z}}{\sigma}[g(\mb z^{(k)}, \mb w^{(k)})-g(\mb z^{(k+1)}, \mb w^{(k+1)})] \\
&=\frac{L_{\mb z}}{\sigma}[(g(\mb z^{(k)}, \mb w^{(k)})-g(\mb x, \mb x))-(g(\mb z^{(k+1)}, \mb w^{(k+1)})-g(\mb x, \mb x))].
\end{align}
Rearranging this last equation and applying it recursively then gives that
\begin{align}
g(\mb z^{(k)}, \mb w^{(k)})-g(\mb x, \mb x) \le \paren{1-\frac{\sigma}{L_{\mb z}}}^k\paren{g(\mb z^{(0)}, \mb w^{(0)})-g(\mb x, \mb x)}. 
\end{align}
An analogous statement obviously also holds for the $\mb w$ sequence. Invoking strong convexity again, we have 
\begin{align}
\frac{\sigma}{2}\|(\mb z^{(k)}, \mb w^{(k)})-(\mb x, \mb x)\|^2 \le g(\mb z^{(k)}, \mb w^{(k)})-g(\mb x, \mb x).
\end{align}
Hence,
\begin{align}
\|(\mb z^{(k)}, \mb w^{(k)})-(\mb x, \mb x)\| \le \paren{1-\frac{\sigma}{L_{\mb z}}}^{k/2}\sqrt{\frac{2}{\sigma}\paren{g(\mb z^{(0)}, \mb w^{(0)})-g(\mb x, \mb x)}}. 
\end{align}
Substituting $\sigma = \frac{1}{3} \norm{\mb x}{}^2$ and $L_{\mb z} = 4\norm{\mb x}{}^2 + \rho$ completes the proof.
\end{proof}

\section{Conclusion and future work}
We have shown that the Gaussian phase retrieval problem can be formulated as a biconvex optimization problem, and that ADMM applied to this problem has a dual variable $\mb \lambda$ equal to zero at all critical points.  This motivated a convergence analysis of the block coordinate descent algorithm, which is equivalent to ADMM when the dual variable $\mb \lambda$ is fixed to be zero.  We established that the block coordinate descent algorithm converges to a global minimizer of the expected objective at a linear rate locally, when starting from a ``close'' initial point---such a ``close'' point can always be found using a spectral method. 

One can expect to show a similar result for the finite-sample objective using a concentration argument. On our specific nonconvex problem, ADMM with dual fixed as zero is equivalent to the BCD method. In general, they are not. Both methods can be notably competitive in performance when solving certain classes of structural large-scale nonconvex optimization problems. Theoretical understanding of their behaviors is largely open.


\section{Acknowledgments}

The authors would like to foremostly thank their research advisor, Professor Emmanuel J. Cand\`{e}s, for introducing and guiding this research.  D.B. is also very grateful to Professor Walter Murray and Professor Gordon Wetzstein for many helpful discussions, etc.

\bibliographystyle{abbrvnat}
\bibliography{nips_bib}

\section{Appendix: Proofs of lemmas}

\begin{proof}[Proof of \cref{lem:Lag-crit}]
The ``if'' part is simple. We next show the ``only if'' part. Taking partial derivatives of $\mc L$ and setting them to be zero, we obtain 
\begin{align}
\partial_{\mb z} f(\mb z, \mb w) + \mb \lambda + \rho (\mb z - \mb w) & = \mb  0, \\
\partial_{\mb w} f(\mb z, \mb w) - \mb \lambda - \rho (\mb z - \mb w) & = \mb  0, \\
\mb z & = \mb w.
\end{align}
Since
\begin{align}
\partial_{\mb z} f(\mb z, \mb w)&=\frac{1}{2m}\sum \limits_{k=1}^m \mb a_k^T \mb w (\mb a_k^T \mb z \mb a_k^T \mb w-y_k^2) \mb a_k, \\
\partial_{\mb w} f(\mb z, \mb w)&=\frac{1}{2m}\sum \limits_{k=1}^m \mb a_k^T \mb z (\mb a_k^T \mb z \mb a_k^T \mb w-y_k^2) \mb a_k,
\end{align}
we have the equality $\partial_{\mb z} f(\mb z, \mb w)=\partial_{\mb w} f(\mb z, \mb w)$ when $\mb z = \mb w$.  Subtracting the second equation from the first then gives $2 \mb \lambda = \mb 0$.  Adding the first two equations gives $\partial_{\mb z} f(\mb z, \mb w)=\partial_{\mb w} f(\mb z, \mb w)= \mb 0$.
%
\end{proof}

\begin{proof}[Proof of \cref{lem:block-Lipschitz}]
By a direct computation,
\begin{align} \label{eqn:gradient}
\nabla_{\mb z}g(\mb z, \mb w) & = 2\mb w^T \mb z \mb w + \norm{\mb w}{}^2 \mb z - 2\mb x^T \mb w \mb x - \norm{\mb x}{}^2 \mb w + \rho\paren{\mb z - \mb w}, \\
\nabla_{\mb w}g(\mb z, \mb w) & = 2\mb w^T \mb z \mb z + \norm{\mb z}{}^2 \mb w - 2\mb x^T \mb z\mb x - \norm{\mb x}{}^2 \mb z - \rho\paren{\mb z - \mb w}. 
\end{align}
Hence, 
\begin{align}
\|\nabla_{\mb z}g(\mb z + \mb h_{\mb z}, \mb w)-\nabla_{\mb z}g(\mb z, \mb w)\| 
&= \|2\mb w^T \mb h \mb w + \| \mb w \|^2 \mb h + \rho \mb h\| \\
& \le \paren{2\norm{\mb w \mb w^T}{} + \norm{\mb w}{}^2 + \rho} \norm{\mb h}{}\\
&= (3\| \mb w \|^2 + \rho) \| \mb h\| \\
& \le (4\|\mb x \|^2+\rho) \| \mb h \|,
\end{align}
where in the last inequality we have used the fact that $\|\mb w - \mb x\| \le \frac{1}{8}\| \mb x \|$, and hence $\| \mb w \| \le \frac{9}{8}\| \mb x \|$.  An entirely analogous argument provides the block Lipschitz constant $L_{\mb w}$.
\end{proof}

\begin{proof}[Proof of Lemma~\ref{lem:strong-cvx}]
A direct computation shows that the Hessian quadratic form is 
\begin{multline}
\begin{bmatrix}
\mb h_{\mb z} \\ \mb h_{\mb w}
\end{bmatrix}^T 
\nabla^2 \expect{g} 
\begin{bmatrix}
\mb h_{\mb z} \\ \mb h_{\mb w}
\end{bmatrix} = \paren{\mb w^T \mb h_{\mb z} + \mb z^T \mb h_{\mb w}}^2 + 2\mb w^T \mb z \mb h_{\mb w}^T \mb h_{\mb z} + \frac{1}{2} \norm{\mb z}{}^2 \norm{\mb h_{\mb w}}{}^2 + \frac{1}{2}   \norm{\mb w}{}^2 \norm{\mb h_{\mb z}}{}^2 \\
+ 2 \mb z^T \mb h_{\mb z} \mb w^T \mb h_{\mb w} - 2 \mb x^T \mb h_{\mb z} \mb x^T \mb h_{\mb w} - \norm{\mb x}{}^2 \mb h_{\mb w}^T \mb h_{\mb z} + \frac{\rho}{2} \norm{\mb h_{\mb z} - \mb h_{\mb w}}{}^2. 
\end{multline}
Write $\mb z = \mb x + \mb \delta_{\mb z}$ and $\mb w = \mb x + \mb \delta_{\mb w}$. Then, $\norm{\mb \delta_{\mb z}}{} \le \frac{1}{8} \norm{\mb x}{}$ and $\norm{\mb \delta_{\mb w}}{} \le \frac{1}{8} \norm{\mb x}{}$ by our assumption. For any $\mb v = [\mb h_{\mb z}; \mb h_{\mb w}]$ with $\norm{\mb v}{} = 1$, 
\begin{align}
& \mb v^T \nabla \expect{g\paren{\mb z, \mb w}} \mb v \nonumber \\
=\; & \begin{bmatrix} 
\mb h_{\mb z} \\ \mb h_{\mb w}
\end{bmatrix}^T 
\nabla^2 \expect{g(\mb x+ \mb \delta_{\mb z}, \mb x+ \mb \delta_{\mb w}} 
\begin{bmatrix}
\mb h_{\mb z} \\ \mb h_{\mb w}
\end{bmatrix} \\
\ge \; & 2\paren{\mb x + \mb \delta_{\mb w}}^T \paren{\mb x + \mb \delta_{\mb z}} \mb h_{\mb w}^T \mb h_{\mb z} + \frac{1}{2} \norm{\mb x + \mb \delta_{\mb z}}{}^2 \norm{\mb h_{\mb w}}{}^2 + \frac{1}{2} \norm{\mb x + \mb \delta_{\mb w}}{}^2 \norm{\mb h_{\mb z}}{}^2 \nonumber \\
& \quad + 2\paren{\mb x^T \mb h_{\mb z} \mb \delta_{\mb w}^T \mb h_{\mb w} + \mb x^T \mb h_{\mb w} \mb \delta_{\mb z}^T \mb h_{\mb z} + \mb \delta_{\mb z}^T \mb h_{\mb z} \mb \delta_{\mb w}^T \mb h_{\mb w}} - \norm{\mb x}{}^2 \mb h_{\mb w}^T \mb h_{\mb z} + \frac{\rho}{2} \norm{\mb h_{\mb z} - \mb h_{\mb w}}{}^2\\
& \qquad (\text{the first square term ignored})\nonumber \\
\ge \; & \frac{1}{2}\norm{\mb x}{}^2\norm{\mb h_{\mb z} + \mb h_{\mb w}}{}^2 + 2\paren{\mb x^T \mb \delta_{\mb z} + \mb x^T \mb \delta_{\mb w} + \mb \delta_{\mb w}^T \mb \delta_{\mb z}} \mb h_{\mb w}^T \mb h_{\mb z} + \mb x^T \mb \delta_{\mb z} \norm{\mb h_{\mb w}}{}^2 + \mb x^T \mb \delta_{\mb w} \norm{\mb h_{\mb z}}{}^2 \nonumber \\
& \quad + 2\paren{\mb x^T \mb h_{\mb z} \mb \delta_{\mb w}^T \mb h_{\mb w} + \mb x^T \mb h_{\mb w} \mb \delta_{\mb z}^T \mb h_{\mb z} + \mb \delta_{\mb z}^T \mb h_{\mb z} \mb \delta_{\mb w}^T \mb h_{\mb w}} + \frac{\rho}{2} \norm{\mb h_{\mb z} - \mb h_{\mb w}}{}^2 \\
& \qquad (\text{expand the norm squared $\norm{\mb x + \mb \delta_{\mb z}}{}^2$ and $\norm{\mb x + \mb \delta_{\mb w}}{}^2$ and complete a square term})\nonumber \\
\ge\; &  \frac{1}{2}\norm{\mb x}{}^2\norm{\mb h_{\mb z} + \mb h_{\mb w}}{}^2 + \frac{\rho}{2} \norm{\mb h_{\mb z} - \mb h_{\mb w}}{}^2 - \frac{17}{16} \norm{\mb x}{}^2 \norm{\mb h_{\mb w}}{} \norm{\mb h_{\mb z}}{} - \frac{1}{8}\norm{\mb x}{}^2 \paren{\norm{\mb h_{\mb w}}{}^2 + \norm{\mb h_{\mb z}}{}^2} \\
& \qquad (\text{by Cauchy-Schwarz inequality and norm bounds on $\mb \delta_{\mb z}$ and $\mb \delta_{\mb w}$})\nonumber \\
\ge\; & \frac{1}{2}\norm{\mb x}{}^2 \paren{\norm{\mb h_{\mb z} + \mb h_{\mb w}}{}^2 + \norm{\mb h_{\mb z} - \mb h_{\mb w}}{}^2} - \frac{21}{32} \norm{\mb x}{}^2 \\
& \qquad (\text{by $\rho \ge \norm{\mb x}{}^2$ and the inequality $2ab \le a^2 + b^2$})\nonumber \\
=\; & \frac{11}{32} \norm{\mb x}{}^2 \ge \frac{1}{3} \norm{\mb x}{}^2,   
\end{align}
completing the proof. 
\end{proof}

\begin{proof}[Proof of \cref{lem:no-escape}]
The explicit formula for the iterates given by \eqref{alg:block-coord-desc} is:
\begin{align}
\mb z^{(k+1)} & = \paren{2\mb w^{k)} \paren{\mb w^{(k)}}^T  + \norm{\mb w^{(k)}}{}^2 \mb I + \rho \mb I}^{-1} \paren{2\mb x^T \mb w^{(k)} \mb x + \norm{\mb x}{}^2 \mb w^{(k)} + \rho \mb w^{(k)}}, \\
\mb w^{(k+1)} & = \paren{2\mb z^{(k+1)} \paren{\mb z^{(k+1)}}^T + \norm{\mb z^{(k+1)}}{}^2 \mb I+ \rho \mb I}^{-1} \paren{2\mb x^T \mb z^{(k+1)} \mb x + \norm{\mb x}{}^2 \mb z^{(k+1)} + \rho \mb z^{(k+1)}}. 
\end{align}
We will show that for any $\paren{\mb z, \mb x} \in N_{\mb x}$, one round of update will produce $\paren{\mb z^{+}, \mb w^{+}}$ that stays in $N_{\mb x}$. We will only show the proof for $\mb z^{+}$; proof for $\mb w^{+}$ is similar. 
Note that $2\mb w \mb w^T + \paren{\norm{\mb w}{}^2 + \rho} \mb I$ can be eigen-decomposed as 
\begin{align}
2\mb w \mb w^T + \paren{\norm{\mb w}{}^2 + \rho} \mb I = \paren{3\norm{\mb w}{}^2 + \rho} \frac{\mb w \mb w^T}{\norm{\mb w}{}^2} + \sum_{p=1}^{n-1} \paren{\norm{\mb w}{}^2 + \rho} \mb v_p \mb v_p^T, 
\end{align}
where $\mb v_p$ are mutually orthogonal unit vectors that are all orthogonal to $\mb w$. Thus, 
\begin{align}
\paren{2\mb w \mb w^T + \paren{\norm{\mb w}{}^2 + \rho} \mb I}^{-1} = \paren{3\norm{\mb w}{}^2 + \rho}^{-1} \frac{\mb w \mb w^T}{\norm{\mb w}{}^2} + \sum_{p=1}^{n-1} \paren{\norm{\mb w}{}^2 + \rho}^{-1} \mb v_p \mb v_p^T. 
\end{align}
Substituting the above inverse formula into the updating equation, we have 
\begin{equation}
\begin{split}
\mb z^{+} &= \paren{3\norm{\mb w}{}^2 + \rho}^{-1}\paren{\norm{\mb x}{}^2 + \rho} \mb w + 2\paren{3\norm{\mb w}{}^2 + \rho}^{-1} \frac{\paren{\mb w^T \mb x}^2}{\norm{\mb w}{}^2}\mb w \\
&+ 2\sum_{p=1}^{n-1} \paren{\norm{\mb w}{}^2 + \rho}^{-1} \mb v_p \mb v_p^T \mb x \mb w^T \mb x. 
\end{split}
\end{equation}
Moreover, we have 
\begin{align}
\norm{\mb x}{}^2 = \norm{\frac{\mb w \mb w^T}{\norm{\mb w}{}^2} \mb x}{}^2 + \norm{\sum_{p=1}^{n-1} \mb v_p \mb v_p^T \mb x}{}^2 = \frac{\paren{\mb w^T \mb x}^2}{\norm{\mb w}{}^2} + \sum_{p=1}^{n-1} \paren{\mb v_p^T \mb x}^2. 
\end{align}
Assume 
\begin{align}
\norm{\mb w}{} & = \beta \norm{\mb x}, \\
\mb w^T \mb x & = \alpha \norm{\mb w}{} \norm{\mb x}{} = \alpha \beta \norm{\mb x}{}^2. 
\end{align}
Then, 
\begin{align}
\sum_{p=1}^{n-1} \paren{\mb v_p^T \mb x}^2 = \paren{1-\alpha^2} \norm{\mb x}{}^2. 
\end{align}
Showing $\norm{\mb z^{+} - \mb x}{} \le 1/8 \cdot \norm{\mb x}{}$ is equivalent to showing $2\innerprod{\mb z^{+}}{\mb x} - \norm{\mb z^{+}}{}^2 \ge 63/64 \cdot \norm{\mb x}{}^2$. Further write $\rho = C \norm{\mb x}{}^2$ for some $C > 0$ to be determined later. Then,  
\begin{align} \label{eq:cd_no_escape_cross}
2\innerprod{\mb z^{+}}{\mb x} = \brac{2\paren{3\beta^2 + C}^{-1} \paren{1 + C+ 2\alpha^2} \alpha \beta + 4\paren{\beta^2 + C}^{-1} \alpha \beta \paren{1-\alpha^2}}\norm{\mb x}{}^2, 
\end{align}
and 
\begin{align}  \label{eq:cd_no_escape_norm}
\norm{\mb z^{+}}{}^2 = \brac{\beta^2 \paren{3\beta^2 +C}^{-2} \paren{1 + C + 2\alpha^2}^2 + 4\alpha^2 \beta^2 \paren{\beta^2 + C}^{-2} \paren{1-\alpha^2} } \norm{\mb x}{}^2. 
\end{align}
By our assumption that $\norm{\mb x - \mb w}{} \le 1/8 \cdot \norm{\mb x}{}$, we have 
\begin{align} \label{eq:cd_no_escape_range1}
7/8 \le \beta \le 9/8,  \quad (\text{triangular inequality}) 
\end{align}
and 
\begin{align} \label{eq:cd_no_escape_range2}
2\alpha \beta - \beta^2 \ge 63/64 \Longrightarrow \beta/2 + 63/\paren{128\beta}\le \alpha \le 1. 
\end{align}
From the above calculation, 
\begin{align}
2\innerprod{\mb z^{+}}{\mb x} - \norm{\mb z^{+}}{}^2 = h\paren{\alpha, \beta, C} \norm{\mb x}{}^2, 
\end{align}
where $h\paren{\alpha, \beta, C}$ can be read off from~\cref{eq:cd_no_escape_cross,eq:cd_no_escape_norm}. Taking partial derivative wrt $C$, we obtain 
\begin{equation}
\begin{split}
\nabla_C h &= -2\paren{3\beta^2 + C}^{-2}\paren{1+C+2\alpha^2} \alpha \beta + 2\paren{3\beta^2+C}^{-1} \alpha \beta -4 \paren{\beta^2 + C}^{-2} \alpha \beta \paren{1 - \alpha^2} \\
&+ 2\beta^2\paren{3\beta^2 + C}^{-3} \paren{1+C+2\alpha^2}^2 - 2\beta^2\paren{3\beta^2 + C}^{-2}\paren{1+C+2\alpha^2} \\
 &+8 \alpha^2 \beta^2 \paren{\beta^2 + C}^{-3} \paren{1-\alpha^2}. 
\end{split}
\end{equation}

When $\alpha = 1, \beta = 1$, $\nabla_C h = 0$ and $h = 1$. 

Otherwise, suppose $C > 1$. We have 
\begin{align}
\nabla_C h & \le 2\paren{3\beta^2 +C}^{-3} \beta \left[-\paren{3\beta^2+C} \paren{1+C+2\alpha^2} \alpha + \paren{3\beta^2 +C}^2 \alpha - 2\paren{3\beta^2 + C} \alpha \paren{1-\alpha^2} \right. \nonumber\\
& \quad \left. + \beta \paren{1+C+2\alpha^2}^2 - \beta\paren{3\beta^2 + C}\paren{1+C+2\alpha^2} + 4\alpha^2 \beta \paren{1-\alpha^2}\right] \\
& =  2\paren{3\beta^2 +C}^{-3} \beta \left[\paren{3\alpha \beta^2 -3\alpha + \beta + 2\alpha^2 \beta - 3\beta^3} C \right. \nonumber \\
& \quad \left. + \paren{9\alpha \beta^4 - 9 \alpha \beta^2 + \beta + 8\alpha^2 \beta - 3\beta^3 - 6\alpha^2 \beta^3}  \right] \\
& = 2\paren{3\beta^2 +C}^{-3} \beta \left\{\brac{-3\paren{\alpha - \beta}^2 \paren{\alpha + \beta} + \paren{\alpha^2 -1} \paren{3\alpha - \beta}} C \right. \nonumber \\
& \quad \left. + \paren{9\alpha \beta^4 - 9 \alpha \beta^2 + \beta + 8\alpha^2 \beta - 3\beta^3 - 6\alpha^2 \beta^3} \right\}. 
\end{align}
When $\alpha$ and $\beta$ do not assume $1$ simultaneously, 
\begin{align}
-3 \paren{\alpha - \beta}^2 \paren{\alpha + \beta} + \paren{\alpha^2 - 1}\paren{3\alpha - \beta} < 0
\end{align}
for $\alpha, \beta$ in the range computed in~\cref{eq:cd_no_escape_range1,eq:cd_no_escape_range2}. Thus, whenever 
\begin{align} \label{eq:cd_no_escape_cbound}
C \ge \frac{9\alpha \beta^4 - 9 \alpha \beta^2 + \beta + 8\alpha^2 \beta - 3\beta^3 - 6\alpha^2 \beta^3}{3\paren{\alpha - \beta}^2 \paren{\alpha + \beta} - \paren{\alpha^2 -1} \paren{3\alpha - \beta}}, 
\end{align}
we have $\nabla_C h \le 0$. Suppose the above condition on $C$ holds. For any fixed $\paren{\alpha, \beta}$, the minimum of $h$ is attained when $C \to \infty$, which implies 
\begin{align}
h \ge 2\alpha \beta - \beta^2 \ge 63/64. 
\end{align}
So, $\mb z^{+}$ remains in $N_{\mb x}$. 

Now we provide an upper bound for the right side of~\cref{eq:cd_no_escape_cbound}, denoted as $R_C$, to complete the proof. We know that the denominator of $R_c > 0$ when $\alpha = 1, \beta = 1$ does not happen. It can also be checked by examining the derivatives that the denominator is monotonically decreasing wrt $\alpha$, and the nominator is monotonically increasing wrt $\alpha$. These imply the maximum for $R_C$ is attained when $\alpha = 1$. So, 
\begin{align}
R_c \le \frac{9\beta^4 - 9\beta^2 + 9\beta -9\beta^3 }{3\paren{1-\beta}^2 \paren{1+\beta}} = 3\beta \le 27/8, 
\end{align}
as claimed. 
\end{proof}

\end{document}